\documentclass[11pt,a4paper]{article}
\usepackage[utf8]{inputenc}
\usepackage[T1]{fontenc}
\usepackage{amsmath,amssymb,amsthm}
\usepackage[a4paper,margin=2.5cm]{geometry}
\usepackage[authoryear,round]{natbib}
\usepackage{graphicx}
\usepackage[colorlinks=true,urlcolor=blue,citecolor=black,linkcolor=black]{hyperref}
\newtheorem{theorem}{Theorem}[section]
\newtheorem{remark}[theorem]{Remark}
\newcommand{\sech}{\operatorname{sech}}
\newcommand{\rd}{\mathrm{d}}
\title{Inverse construction of exact atmospheric models and the 3.7--4.4\,mHz acoustic resonances}
\author{O.\,V.~Kaptsov\\[2pt]
\small Federal Research Center for Information and Computational Technologies, Novosibirsk, Russia\\
\small \texttt{profkap@gmail.com}, ORCID 0000-0002-9562-9092}
\date{}

\begin{document}
\maketitle

\begin{abstract}
Large explosive eruptions and earthquakes excite atmospheric oscillations near 3.7 and 4.4\,mHz that couple to
the solid Earth. For vertically propagating linear acoustic--gravity waves in a hydrostatic ideal-gas atmosphere,
the travel-time formulation gives a Klein--Gordon equation with a potential $Q(\tau)$. We establish three main results.
First, for every potential admitting a positive zero-energy solution, the corresponding atmospheric profiles are
constructed explicitly by quadrature; this extends the known reflectionless profiles of constant $Q$. Second,
Darboux--Crum potentials give a smooth atmosphere that reproduces the U.S. Standard Atmosphere between 11 and
100\,km to within 3\,K and is exactly transparent to vertical waves above its cut-off; it supports no damped overtone near 4.4\,mHz,
and its ground-confined mode below the cut-off is perfectly trapped. The observed spectrum therefore requires a
cut-off barrier of finite thickness connecting the lower atmosphere to the thermosphere. Third, a Rosen--Morse barrier gives an exactly solvable atmosphere with a closed-form
hypergeometric resonance condition. Its resonances at 3.76 and 4.37\,mHz, with leakage quality factors 173 and 20,
explain the different damping of the two oscillations: the fundamental lies below the top of the mesospheric
cut-off barrier and leaks by tunnelling, whereas the overtone lies above it. Full-wave computations for
event-specific atmospheres reproduce the reported 3.7--3.8\,mHz fundamental oscillation in all six events, of
seven considered, in which such a mode was identified.
\end{abstract}

\section{Introduction}\label{sec:intro}

The 1991 eruption of Mount Pinatubo produced harmonic Rayleigh waves with a bichromatic spectrum: a
dominant peak at 4.44\,mHz and a weaker one at 3.68\,mHz \citep{KM1992,WZ1992}. The 1982 eruption of
El Chich\'on produced peaks at 3.70 and 5.14\,mHz \citep{WZ1992}. The 2022 eruption of Hunga
Tonga--Hunga Ha'apai excited predominantly the spheroidal mode ${}_0S_{29}$ (3.72\,mHz), about eleven times
more strongly than Pinatubo, and near-field barometers recorded the atmospheric oscillation itself at
3.71--3.73\,mHz with a quality factor of about 112 \citep{Ringler2023}. The accepted interpretation
\citep{KMH1994,Lognonne1998,WK2010} is that acoustic modes trapped between the ground and the lower
thermosphere are excited near the source and couple to Rayleigh waves of the same frequency and horizontal
wavelength. Numerical normal-mode computations for a spherical Earth with an atmosphere give the
fundamental and first overtone near 3.68 and 4.40\,mHz with quality factors 115 and 21
\citep{Lognonne1998}; \citet{Tahira1995} modelled the fundamental as a standing wave with an
evanescent upper boundary and a wavelength about four times the height of the mesopause.

Independently, there is a line of work on reflectionless propagation of waves in strongly inhomogeneous
media \citep{PPB2011,PPT2012,PPT2023}. For vertically propagating acoustic waves
in a gravitating atmosphere, a change of variables reduces the wave equation to a Klein--Gordon equation
with constant coefficients provided the temperature profile satisfies a nonlinear ordinary differential
equation, whose integration gives a family of transparent atmospheres \citep{PPT2012}.

This paper connects the two topics. Its purpose is not a complete global Earth--atmosphere normal-mode
model, but to isolate the vertical atmospheric resonator and to determine analytically which part of the
observed spectrum it controls; sphericity, winds, variable composition and coupling to the elastic Earth are
left aside. We construct atmospheres explicitly from the potential (\S\,\ref{sec:inverse}), show that a transparent
atmosphere reproducing the middle atmosphere has a perfectly trapped fundamental but no leaky overtone
(\S\,\ref{sec:transparent}), solve exactly a model
with a mesospheric barrier (\S\,\ref{sec:resonance}), estimate corrections (\S\,\ref{sec:corrections}), and compare
with event-specific computations and observations (\S\S\,\ref{sec:events}--\ref{sec:discussion}). The central message is
that the vertical acoustic spectrum is controlled by the scattering structure of the atmospheric cut-off, not simply
by the presence of a realistic temperature stratification.

\section{Formulation}\label{sec:form}

Let $z$ point upwards, $p_0'=-\rho_0 g$, $c^2=\gamma p_0/\rho_0=\gamma RT/M$. The linearised inviscid equations are
\begin{equation}\label{eq:lin}
\rho_0\boldsymbol{v}_t=-\nabla p-\rho g\,\boldsymbol{e}_z,\qquad \rho_t+\nabla\cdot(\rho_0\boldsymbol{v})=0,\qquad
p_t+w\,p_0'=c^2(\rho_t+w\rho_0').
\end{equation}
For horizontally uniform motion, elimination of $\rho$ and $p$ gives exactly, within the linearised inviscid model,
$\rho_0w_{tt}=(\rho_0c^2w_z)_z$, i.e.
\begin{equation}\label{eq:vert}
w_{tt}=c^2(z)\,w_{zz}-\gamma g\,w_z ,
\end{equation}
with $p_t=-\rho_0c^2w_z+\rho_0gw$. With the travel time $\tau=\int\rd z/c$ (dots denote $\rd/\rd\tau$) and
$a=(\dot c+\gamma g)/c$, the substitution $w=u\exp(\frac12\int a\,\rd\tau)$ yields
\begin{equation}\label{eq:KG}
u_{tt}=u_{\tau\tau}-Q(\tau)\,u,\qquad Q=\tfrac14a^2-\tfrac12\dot a .
\end{equation}
For an isothermal atmosphere $Q=\omega_a^2=(\gamma g/2c)^2$ is the squared acoustic cut-off frequency; in
general $Q(\tau)^{1/2}$ provides a local cut-off scale, and $Q$ governs reflection and trapping.

\section{Reconstruction of the atmosphere from the potential}\label{sec:inverse}

Throughout this section $g$ and $\gamma$ are constant (the variation of $g$ is below 3\,\% beneath 100\,km and is
included in the numerical computations of \S\S\,\ref{sec:resonance}--\ref{sec:events}), and an atmosphere is a
hydrostatic ideal-gas profile with sound speed $c\in C^2(J)$, $c>0$, on an interval $J$ of heights. Since $c>0$,
$\tau(z)=\int\rd z/c$ is a strictly increasing $C^3$ map of $J$ onto an interval $I$, and $c$ may be regarded as a
function of $\tau\in I$.

\begin{theorem}\label{thm:inverse}
(i) Let $Q\in C(I)$, let $\psi\in C^2(I)$ be a positive solution of $\ddot\psi=Q\psi$, let $\tau_*\in I$ be a point
at which $\int_\tau^{\tau_*}\psi^2\,\rd\tau'$ is finite for all $\tau\in I$, and let $C\in\mathbb{R}$ be such that
\begin{equation}\label{eq:inverse}
c(\tau)=\frac{C+\gamma g\int_\tau^{\tau_*}\psi^2\,\rd\tau'}{\psi^2(\tau)}>0\quad(\tau\in I),\qquad z(\tau)=\int c\,\rd\tau .
\end{equation}
Then $c\in C^2(I)$, $z$ is a strictly increasing map of $I$ onto an interval $J$, the potential (\ref{eq:KG}) of the
atmosphere $c(z)$ equals $Q$, and $w=u/\psi$ solves (\ref{eq:vert}) for every solution $u$ of (\ref{eq:KG}).

(ii) Conversely, for every atmosphere in the above sense with $c\in C^2$, the function
$\psi=\exp(-\frac12\int a\,\rd\tau)$ is positive and satisfies $\ddot\psi=Q\psi$, and $c$ is given by (\ref{eq:inverse}) with this
$\psi$, any $\tau_*\in I$ and $C=c(\tau_*)\psi^2(\tau_*)$.

(iii) The overall scale of $\psi$ is immaterial, since the simultaneous transformation $\psi\mapsto A\psi$,
$C\mapsto A^2C$ leaves $c(\tau)$ unchanged. For a given $Q$, the admissible atmospheres are parametrised by the positive
solutions of $\ddot\psi=Q\psi$, modulo this scaling, together with the integration constant $C$.
\end{theorem}

\noindent\textit{Proof.} The gauge factor of \S\,\ref{sec:form} equals $1/\psi$ precisely when $a=-2\dot\psi/\psi$, and
then $\frac14a^2-\frac12\dot a=\ddot\psi/\psi=Q$. The relation $(\dot c+\gamma g)/c=-2\dot\psi/\psi$ is the linear equation
$(c\psi^2)\dot{}=-\gamma g\psi^2$, whose solutions are exactly (\ref{eq:inverse}); $c\in C^2$ because $\psi\in C^2$ and $\psi>0$, and
$z$ is strictly increasing because $c>0$. For (ii), $a=(\dot c+\gamma g)/c\in C^1(I)$, so $\psi\in C^2$, $\psi>0$, and
the same linear equation holds. Statement (iii) follows from (i), (ii) and the scaling invariance; positive solutions of
$\ddot\psi=Q\psi$, when they exist, form a convex cone in a two-dimensional space. \hfill$\square$

\begin{remark}
For constant $Q$, (\ref{eq:inverse}) gives in explicit parametric form the reflectionless profiles found implicitly
by \citet{PPT2012}: for $Q=0$, $\psi=1$ gives $c^2=C^2-2\gamma gz$, and $\psi=\tau$ with $C=0$ gives $c=-\gamma g\tau/3$,
$z=z_*-\gamma g\tau^2/6$, i.e.\ $c^2=\frac23\gamma g(z_*-z)$. The constant $C$ changes the atmosphere but not $Q$. A positive
zero-energy solution exists, for instance, whenever $Q\ge0$ (the solution with $\dot\psi(\tau_0)=0$ is then convex) and
for the Darboux--Crum potentials below, which covers all potentials used here.
\end{remark}

\section{Exactly transparent atmospheres}\label{sec:transparent}

Darboux--Crum transformations \citep{Crum1955} of $Q_0=\omega_a^2$ with seeds $f_j''=\kappa_j^2f_j$,
$0<\kappa_j<\omega_a$, give the reflectionless potentials $Q=\omega_a^2-2\partial_\tau^2\ln W(f_1,\dots,f_n)$ with the
positive $\psi=W(f_1,\dots,f_n,e^{-\omega_a\tau})/W(f_1,\dots,f_n)$. For $n=1$, $f_1=\cosh\kappa\tau$, $C=0$ the
quadrature is elementary: with $c_\infty=\gamma g/2\omega_a$ and up to an additive constant in $z$,
\begin{equation}
z=c_\infty\Big[\tau-\frac{1}{\omega_a+\kappa\tanh\kappa\tau}\Big],\qquad
\frac{c}{c_\infty}=1+\frac{\kappa^2\sech^2\kappa\tau}{(\omega_a+\kappa\tanh\kappa\tau)^2},
\end{equation}
an isothermal atmosphere with a warm layer, $T_{\max}/T_\infty=(1-\kappa^2/\omega_a^2)^{-2}$, transparent at all
frequencies above $\omega_a$ and supporting a trapped oscillation at $\omega^2=\omega_a^2-\kappa^2$. A two-soliton
potential ($f_1=\cosh\kappa_1\tau$, $f_2=\sinh\kappa_2(\tau-\tau_2)$, $C>0$) reproduces the U.S. Standard Atmosphere
\citep{US76} between 11 and 100\,km with an r.m.s.\ error of 3.2\,K ($T_\infty=166$\,K,
$\kappa_1/\omega_a=0.447$, $\kappa_2/\omega_a=0.625$) and a full-wave solution confirms its transparency ($|R|^2<10^{-7}$).

The middle atmosphere can thus be represented to within a few kelvin by a profile that does not reflect vertical
waves above its cut-off (4.23\,mHz), and its response to ground motion has no damped overtone near 4.4\,mHz. Below the
cut-off, with a rigid ground, it has a ground-confined mode at 3.745\,mHz that is perfectly trapped, because $Q\to\omega_a^2$
at the top (the constant $C$ heats the upper layers but leaves $Q$ unchanged), so no energy can leak. Reflectionless
potentials do support trapped modes; what the observed spectrum, a leaky fundamental and a strongly damped overtone,
requires in the present setting is a cut-off barrier of finite thickness above which the cut-off decreases, i.e.\ a
non-reflectionless barrier connecting the lower atmosphere to the thermosphere. The relevant control parameter is
not the realism of the temperature profile by itself, but the scattering structure of the acoustic cut-off.

\section{An exactly solvable model of the resonances}\label{sec:resonance}

The response of the U.S. Standard Atmosphere to vertical ground motion, with a radiation condition at the
top, peaks at 3.75 and 4.48\,mHz; replacing the lowest 50\,km by an isotherm changes these by less than 0.08\,mHz,
whereas removing the cold mesopause or the thermospheric rise shifts them by 0.1--0.2\,mHz. In terms of (\ref{eq:KG}), the cold mesosphere
produces a maximum of $Q$, a barrier for waves below the local cut-off that separates the region adjacent to
the ground from the hot thermosphere. We therefore take the Rosen--Morse potential
\begin{equation}\label{eq:RM}
Q(\tau)=q_0+q_1\tanh(\tau/T_w)+q_2\sech^2(\tau/T_w),\qquad q_2>0,
\end{equation}
with $\psi$ the solution of $\ddot\psi=Q\psi$ decaying as $\tau\to\infty$, and fit the constants, $C$ and the position of
the ground to the standard atmosphere between 0 and 200\,km. This gives $q_0=3.148\times10^{-4}$,
$q_1=-1.634\times10^{-4}$, $q_2=3.365\times10^{-4}$\,s$^{-2}$, $T_w=73.7$\,s, ground at $\tau_g=-372.9$\,s
($\psi(0)=1$, $C=0.117$ in SI units, $\tau=0$ at $z=112.2$\,km). Above 80\,km the model follows the standard
atmosphere within about 15\,K, including the mesopause and the thermospheric rise from 190 to 870\,K
(figure~\ref{fig:model}$a$); below 50\,km it is nearly isothermal (245\,K). The barrier top corresponds to
4.12\,mHz and the thermospheric plateau to 1.96\,mHz.

\begin{figure}
\centering
\includegraphics[width=0.9\textwidth]{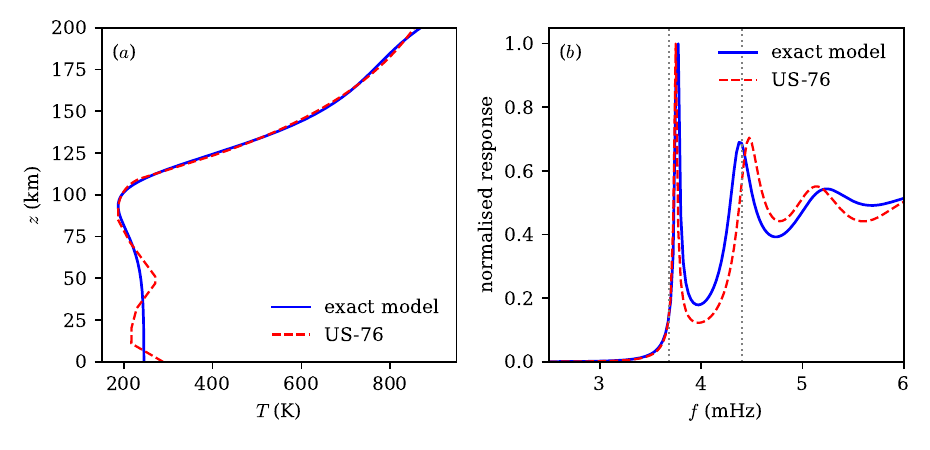}
\caption{($a$) Temperature of the exactly solvable atmosphere (\ref{eq:RM}), (\ref{eq:inverse}) and of the U.S.
Standard Atmosphere. ($b$) Response to vertical ground motion, normalised to its maximum; dotted lines mark
3.68 and 4.40\,mHz.}
\label{fig:model}
\end{figure}

With $x=\tau/T_w$, (\ref{eq:KG}) for $u=e^{-\mathrm{i}\omega t}U(x)$ becomes $U_{xx}+(A-B\tanh x+C'\sech^2x)U=0$,
$A=T_w^2(\omega^2-q_0)$, $B=T_w^2q_1$, $C'=-T_w^2q_2$, and the solution radiating upwards is
\begin{equation}\label{eq:out}
U=y^{\varpi}(1-y)^{\mu/2}\,{}_2F_1\big(\varpi+\tfrac\mu2-\sigma,\,\varpi+\tfrac\mu2+\sigma+1;\,\mu+1;\,1-y\big),
\qquad y=\tfrac12(1+\tanh x),
\end{equation}
with $\mu=-\mathrm{i}T_w(\omega^2-q_0-q_1)^{1/2}$, $\varpi=\frac12(-A-B)^{1/2}$, $\sigma(\sigma+1)=C'$. A rigid ground,
$w=0$, is $U(x_g)=0$ because $w=u/\psi$, so the resonances, i.e.\ the complex poles (quasi-normal modes), are the complex zeros of
\begin{equation}\label{eq:rescond}
F(\omega)={}_2F_1\big(\varpi+\tfrac\mu2-\sigma,\,\varpi+\tfrac\mu2+\sigma+1;\,\mu+1;\,1-y_g\big),\qquad
y_g=\tfrac12\big(1+\tanh(\tau_g/T_w)\big).
\end{equation}
The two lowest zeros, with the quality factor defined as $Q_{\rm leak}={\rm Re}\,\omega/(2|{\rm Im}\,\omega|)$, are
$f=3.763$\,mHz with $Q_{\rm leak}=173$ and $f=4.370$\,mHz with
$Q_{\rm leak}=20$; the real-frequency response from (\ref{eq:out}) coincides with the direct numerical solution and has
maxima at 3.764 and 4.386\,mHz (figure~\ref{fig:model}$b$). For the broad overtone the response maximum need not
coincide with the real part of the pole. Throughout, `resonance' denotes a complex pole, `trapped mode' a real
eigenfrequency without leakage, and `response maximum' a peak of the real-frequency response. These values may be compared with 3.68 and 4.40\,mHz and quality factors 115 and 21
from normal-mode computations for a realistic atmosphere \citep{Lognonne1998}, and with the observed
atmospheric fundamental at 3.71--3.73\,mHz with quality factor 112 \citep{Ringler2023}.

The model makes the mechanism explicit (figure~\ref{fig:mech}). The fundamental lies below the top of the barrier (4.12\,mHz): it is
a quasi-bound state between the ground and the mesospheric barrier, whose energy escapes only by tunnelling
into the thermosphere, hence its large quality factor. The overtone lies above the barrier top and is confined
only by partial reflection from the barrier and the thermospheric gradient; its quality factor of about 20
agrees with the normal-mode value of \citet{Lognonne1998}, which was obtained without this interpretation. The strong damping of the 4.4\,mHz mode, noted by \citet{Lognonne1998} and \citet{Ringler2023}, and the
weak damping of the fundamental are therefore a consequence of the two modes occupying different scattering
regimes relative to the mesospheric barrier.

\begin{figure}
\centering
\includegraphics[width=0.5\textwidth]{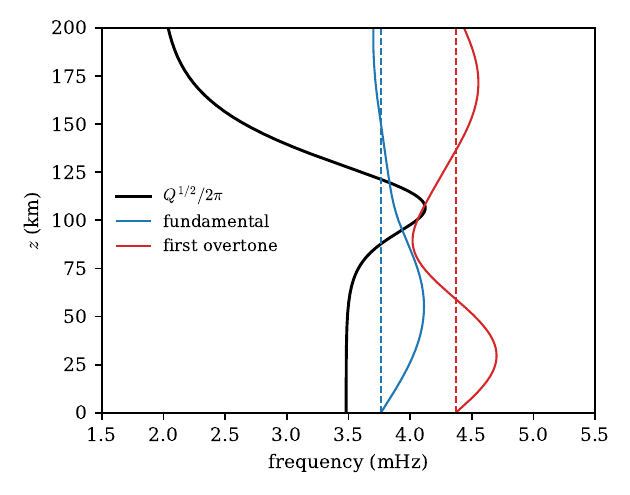}
\caption{Local cut-off frequency $Q^{1/2}/2\pi$ of the exactly solvable atmosphere as a function of height
(black), the real parts of the two lowest resonances (dashed) and the corresponding eigenfunctions $u$ (solid,
arbitrary scale, offset to their frequencies). The fundamental lies below the mesospheric barrier near
105\,km and is evanescent across it; the overtone lies above the barrier top.}
\label{fig:mech}
\end{figure}

\section{Horizontal wavenumber and molecular dissipation}\label{sec:corrections}

For the horizontal wavenumber of ${}_0S_{29}$, $k=\sqrt{l(l+1)}/R_\oplus\approx4.6\times10^{-6}$\,m$^{-1}$, the full system
(\ref{eq:lin}) for the model atmosphere shifts the real-frequency response maxima by $+0.002$ and $+0.003$\,mHz
(computed on a domain truncated at 220\,km, where the $k=0$ maxima lie at 3.763 and 4.394\,mHz), consistent with $\delta\omega\approx k^2\langle c^2(1-N^2/\omega^2)\rangle/2\omega$. The vertical
theory is therefore adequate. Weighting the classical absorption rate
$\delta=\omega^2(\frac43\mu_v+(\gamma-1)\kappa_T/c_p)/(2\rho_0c^2)$ (Sutherland's law for $\mu_v$, the standard-atmosphere law for
$\kappa_T$) with the modal energy density $E=\frac12\rho_0|w|^2+\frac12|p|^2/(\rho_0c^2)$ gives, up to 200\,km,
$Q_{\rm visc}\approx3.6\times10^3$ for the fundamental (94\,\% of its energy below 100\,km) and $4\times10^2$ for the
overtone (27\,\% above 100\,km), so that the total quality factors are 165 and 19. Leakage, not molecular dissipation, thus
controls both modes. The difference between 165 and the value 112 observed near Hunga Tonga is a limitation of
the present model; it may reflect losses not included here (energy transfer to the solid Earth, radiative
relaxation) or the horizontal spreading of the oscillation near a finite source.

\section{Event-specific atmospheres and a WKB formula}\label{sec:events}

Temperature and mean molecular mass were computed from NRLMSIS~2.1 \citep{Emmert2021,Emmert2022}, whose temperature and
major-species components are those of NRLMSIS~2.0, for the time and place
of seven events, with the observed solar and geomagnetic indices from the CelesTrak compilation (e.g.\
$F_{10.7}=200.9$, $Ap=13$ for Pinatubo; 110.2, 20 for Hunga Tonga), $c^2=\gamma RT/M(z)$ and height-dependent $g$.
The full-wave response maxima are compared with observations in table~\ref{tab:events}.

\begin{table}
\begin{center}\small
\begin{tabular}{@{}l@{\hspace{6pt}}p{3.1cm}@{\hspace{6pt}}p{6.3cm}@{}}
Event & computed (mHz) & observed (mHz)\\[2pt]
Pinatubo 1991 & 3.705, 4.395, 5.04 & 3.68, 4.44 (R) \citep{KM1992,WZ1992}\\
El Chich\'on 1982 & 3.71, 4.43, 5.06 & 3.70, 5.14 (R) \citep{WZ1992}\\
Hunga Tonga 2022 & 3.70, 4.42, 5.10 & 3.71--3.73 (P) \citep{Ringler2023}\\
Calbuco 2015 & 3.72, 4.40 & 3.7, 4.4 (I) \citep{NH2015}\\
Tohoku 2011 & 3.71, 4.38, 5.10 & 3.7--3.8, 4.35--4.48 (I) \citep{Rolland2011};\\
 & & 3.7, 4.5, 5.3 (I) \citep{Saito2011}\\
Sumatra 2004 & 3.72, 4.42, 5.09 & 3.9--4.1 (I) \citep{Choosakul2009}; 4.4--4.6 (M) \citep{Iyemori2005}\\
Kelud 2014 & 3.69, 4.39, 5.01, 6.03 & 3.7, 4.6--4.8, 6.7--6.8 (I, R) \citep{Nakashima2016}\\
\end{tabular}
\caption{Computed response maxima for event-specific NRLMSIS atmospheres and reported frequencies.
R: Rayleigh waves; P: near-field surface pressure; I: ionospheric total electron content; M: geomagnetic field.}
\label{tab:events}
\end{center}
\end{table}

The computed fundamental (3.69--3.72\,mHz) reproduces the reported 3.7--3.8\,mHz oscillation in all six events in
which such a mode was identified, for events differing in type, latitude and solar activity; for Sumatra 2004 the
TEC period of about 4\,min lies between the two lowest modes, while the geomagnetic period (4.4--4.6\,mHz) is
compatible with the overtone. The overtone agrees with Pinatubo, Calbuco and Tohoku, and the second overtone with
El Chich\'on. \citet{Rolland2011} observed that after the Tohoku earthquake the 3.7--3.8\,mHz oscillation lasted longer
than the 4.35--4.48\,mHz one, as expected for modes below and above the barrier. The 4.6--4.8 and 6.7--6.8\,mHz peaks
of Kelud are not produced by the vertical resonator; full-wave computations suggest oblique waves of horizontal
wavelength 120--230\,km reflected in the upper thermosphere, which a viscous two-dimensional model would have to
confirm. Varying $F_{10.7}$ from 80 to 280 and $Ap$ from 8 to 40 changes the frequencies by at most 0.01\,mHz: within
the tested range of NRLMSIS parameters the resonance frequencies are set by the lower and middle atmosphere, not by solar
activity. The exactly solvable atmosphere fitted to these profiles reproduces the full-wave values to within 1.5\,\%.

Since the fundamental lies below the barrier top, a WKB quantisation between the rigid ground and the lower
turning point $\tau_t$, $Q(\tau_t)=\omega^2$, applies:
\begin{equation}\label{eq:wkb}
\int_{\tau_g}^{\tau_t}\big(\omega^2-Q(\tau)\big)^{1/2}\rd\tau=\tfrac34\pi .
\end{equation}
It gives 3.763\,mHz for the standard-atmosphere model (pole 3.763\,mHz) and 3.750 and 3.752\,mHz for the NRLMSIS
fits (exact-model response maxima 3.760 and 3.765\,mHz). Replacing $Q$ in the cavity by its lower plateau $\omega_c^2$,
\begin{equation}\label{eq:simple}
f_0\approx\frac{1}{2\pi}\Big[\omega_c^2+\Big(\frac{3\pi}{4\tau_c}\Big)^2\Big]^{1/2},
\end{equation}
where $\tau_c\approx290$\,s is the travel time from the ground to the turning point in the lower mesosphere and
$\omega_c/2\pi\approx3.5$\,mHz, gives 3.72\,mHz in all cases. This is the precise form of the quarter-wave
picture of \citet{Tahira1995}: the fundamental is the cut-off of the lower atmosphere raised by the quantisation of
the cavity, and $\delta f_0/f_0\approx-0.44\,\delta T/T$ for a change of the mean temperature below the barrier.

\section{Discussion}\label{sec:discussion}

The atmospheric branches are almost flat in $k$ (from the full system, the fundamental is 3.7631, 3.7649, 3.7795
and 3.8200\,mHz at $l=0$, 29, 87 and 174), with group velocities of 15--27\,m\,s$^{-1}$ against a horizontal phase
velocity of about 5\,km\,s$^{-1}$. The selection of the excited spheroidal modes is therefore a selection by
frequency, and the free atmospheric oscillation stays within a few hundred kilometres of the source, consistent
with its detection only by nearby barometers \citep{Ringler2023}; the coupling itself would require matching
(\ref{eq:out}) to a solid-Earth solution. The rigid ground condition is accurate for land and ocean, whose surface
admittances ($\sim10^{-7}$ and $\approx2\times10^{-6}$\,m\,s$^{-1}$\,Pa$^{-1}$) are at least $10^3$ times smaller than
$1/(\rho_0c)$, but the frequencies rise by about 0.005 (fundamental) and 0.012\,mHz (overtone) per kilometre of
surface elevation. Which peak dominates a recorded spectrum (4.44\,mHz for Pinatubo, 5.14\,mHz for El Chich\'on) depends also on the
source and on the matching to solid-Earth modes, not described here.

\section{Conclusions}

For vertically propagating acoustic--gravity waves the reconstruction of the atmosphere from the
Klein--Gordon potential is explicit and is obtained by quadrature for every potential admitting a positive
zero-energy solution, so that exactly solvable potentials yield model atmospheres. A reflectionless potential gives
a smooth atmosphere that reproduces the standard middle atmosphere and is transparent above its cut-off, but it has
no damped overtone and its fundamental is perfectly trapped; within the vertical description, the observed spectrum therefore
requires a non-reflectionless cut-off barrier connecting the lower atmosphere to the thermosphere, and is controlled
by the scattering structure of the cut-off rather than by the realism of the temperature profile alone. A
Rosen--Morse barrier gives an atmosphere close to the standard one above 80\,km whose resonances satisfy a
closed-form hypergeometric equation. Its resonances at 3.76 and 4.37\,mHz, with leakage quality factors 173 and 20,
account for the very different damping of the two observed oscillations by tunnelling through, versus reflection
above, the mesospheric cut-off barrier. With event-specific atmospheres the fundamental agrees with the observations in all six
events in which a 3.7--3.8\,mHz atmospheric oscillation was identified, the longer duration of the 3.7\,mHz oscillation after the Tohoku earthquake is
consistent with this picture, the frequencies are insensitive to solar activity within the tested range of NRLMSIS parameters, and a
WKB condition gives the fundamental in terms of the cut-off of the lower atmosphere and the depth of the cavity.

\section*{Acknowledgements}
The author used Claude (Anthropic), an AI assistant, for symbolic and numerical computations, literature search and manuscript preparation; all results were checked by the author, who takes full responsibility for the content.

\section*{Funding}
The research was carried out within the state assignment of Ministry of Science and Higher Education of the Russian Federation for Federal Research Center for Information and Computational Technologies.

\section*{Declaration of interests}
The author reports no conflict of interest.

\section*{Data availability statement}
The U.S. Standard Atmosphere \citep{US76}, NRLMSIS~2.1 \citep{Emmert2021,Emmert2022} and the space-weather indices are publicly available; the scripts are available from the author on request.

\appendix
\section{Numerical methods}\label{app:numerics}

The decaying solution $\psi$ was obtained from the Riccati equation $\dot y=Q-y^2$ for $y=\dot\psi/\psi$, integrated
backwards (the stable direction) with $\Delta\tau=0.1$\,s, and (\ref{eq:inverse}) was evaluated by the trapezoidal
rule; fits used least squares on 1\,km samples. Full-wave solutions used the regular variables
$(\rho_0^{1/2}w,\rho_0^{-1/2}p)$, an eighth-order Runge--Kutta method (relative tolerance $10^{-9}$) and, at the top
(480\,km), the eigenvector of the frozen-coefficient system with upward flux or decay; energy-flux conservation holds
to $10^{-6}$. Complex zeros of (\ref{eq:rescond}) were located by Muller iteration with arbitrary-precision ${}_2F_1$
($|F|<10^{-21}$); (\ref{eq:out}) satisfies its equation to $10^{-30}$.

\end{document}